\documentclass[conference]{IEEEtran}
\IEEEoverridecommandlockouts

\usepackage{cite}
\usepackage{amsmath,amssymb,amsfonts} 
\usepackage{algorithm}
\usepackage{algpseudocode}
\usepackage{graphicx}
\usepackage{textcomp}
\usepackage{xcolor}  
\usepackage{booktabs}
\usepackage{caption}
\usepackage{array}     
\usepackage{multirow}
\usepackage{makecell} 
\usepackage{verbatim}
\usepackage{url} 
\usepackage{ulem}
\usepackage[caption=false,font=normalsize,labelfont=sf,textfont=sf]{subfig}

\newtheorem{theorem}{Theorem}  

\def\BibTeX{{\rm B\kern-.05em{\sc i\kern-.025em b}\kern-.08em
    T\kern-.1667em\lower.7ex\hbox{E}\kern-.125emX}}
\begin{document}

\title{{MARL-Based Sequential RIS Auctions: A Physical-Layer Security Analysis}
}

\author{Yuanyu Zhang,~\IEEEmembership{Member,~IEEE},
	Yu~Zhang, 
	Jialu~He,
	Zhixin~Huang,
	Shuangrui~Zhao,~\IEEEmembership{Member,~IEEE},\\
	and Yulong~Shen,~\IEEEmembership{Member,~IEEE}
}

\maketitle

\begin{abstract}
Reconfigurable intelligent surfaces (RISs) hold great potential to enhance coverage, spectral efficiency, and communication security by intelligently configuring their reflecting elements. When owned by a neutral RIS operator, these elements can be offered as resources for which legitimate receivers and eavesdroppers compete. This paper investigates such competition and evaluates its impact on the physical-layer security performance of legitimate receivers. To model the competition, we develop a sequential RIS auction (SRA) framework, in which a bundle of RIS elements is auctioned in each round through a first-price sealed-bid mechanism, with each bidder submitting its bid based on the achievable rate gain and remaining budget. We then formulate the sequential bidding process as a Markov game by specifying its states, actions, rewards, and state transitions. To solve the game, we propose a multi-bidder deep deterministic policy gradient (MADDPG)-based multi-bidder reinforcement learning (MARL) approach under centralized training and decentralized execution (CTDE), enabling legitimate receivers and eavesdroppers to learn bidding strategies that maximize their long-term economic surplus. Numerical results show that, under the considered eavesdropper bidding strategies, the RL-based strategy enables legitimate receivers to achieve the highest secrecy rate per unit cost, outperforming random and fixed strategies and approaching the ideal physical-layer upper bound.
\end{abstract}


\section{Introduction}

Reconfigurable intelligent surfaces (RISs) can intelligently reshape wireless propagation by configuring their reflecting elements, thereby improving coverage, spectral efficiency, and communication reliability. These capabilities make RISs particularly promising for sixth-generation (6G) networks, which face increasingly complex security threats due to open access, heterogeneous interconnection, and intelligent deployment. Accordingly, RISs have been widely investigated as an enabler of physical-layer security (PLS) to enhance legitimate transmission and suppress information leakage~\cite{Asif2024,Zhang2024,Khoshafa2024,Li2025,Chen2026b,Yang2026a,Salman2026a,Deraz2026,Rasheed2026}. However, when an RIS is owned by a neutral operator and its reflecting elements are offered as scarce resources, both legitimate receivers and eavesdroppers may compete for them. In this case, RIS-assisted security depends not only on resource configuration, but also on the RIS resource allocation between participants with conflicting security objectives.


Recent studies have introduced market-based mechanisms for leasing RIS resources to multiple users through pricing, game-theoretic, and auction-based approaches~\cite{Gao2022,Cai2023,Pan2024,Schwarz2024,Zan2026a,Tang2026a,Tang2026b}. These studies generally assume that all bidders are legitimate users. In practice, however, a neutral RIS provider may be unable to identify a bidder's true security intent, allowing an eavesdropper to participate by posing as a legitimate service requester. Moreover, existing market-based studies mainly evaluate communication performance and economic efficiency, without examining how competition between legitimate receivers and eavesdroppers affects PLS. Therefore, the security implications of market-based RIS resource allocation remain largely unexplored.

To address these gaps, this paper investigates the competition between a legitimate receiver and an eavesdropper for RIS resources provided by a neutral operator. We develop a sequential RIS auction (SRA) framework to characterize their strategic interactions and analyze how auction outcomes affect the PLS performance of the legitimate receiver. Since bidding decisions are coupled across auction rounds through resource allocation and budget consumption, we further model the competition as a Markov game and develop a MARL-based approach to learn long-term bidding strategies. The main contributions are summarized as follows:

\begin{itemize}
\item We introduce a market-based RIS threat model in which an eavesdropper may pose as a legitimate requester and compete for RIS resources provided by a neutral operator. Based on this model, we develop an SRA framework that divides the RIS into bundles of reflecting elements and allocates them round by round through first-price sealed-bid auctions.

\item We formulate the sequential bidding process as a budget-coupled Markov game by defining its states, actions, rewards, and state transitions. To solve the game, we develop a MADDPG-based MARL approach under centralized training and decentralized execution, enabling bidders to learn long-term, budget-aware strategies based on achievable rate gains and remaining budgets.

\item We evaluate the impact of auction competition on the PLS performance of the legitimate receiver. Numerical results show that, under the considered eavesdropper bidding strategies, the learned policy achieves the highest secrecy rate per unit bid among the compared methods, outperforming random and fixed strategies and approaching the ideal physical-layer upper bound.
\end{itemize}
 
The organization of the remainder of this paper is as follows. 
Section~\ref{sec_sys} presents the system model and threat model. 
Section~\ref{sec_SRA} describes the SRA framework. 
Section~\ref{sec_mecha} develops the MARL-based bidding strategy. 
Section~\ref{sec_simul} presents the numerical evaluation results and corresponding analysis.
Section~\ref{sec_con} concludes this paper.

\section{System and Threat Model}\label{sec_sys}

\subsection{Network Model}\label{subsection_sys_net}

We consider an auction-driven RIS-assisted network, which consists of a legitimate transmitter (Alice), a passive RIS relay, an intended receiver (Bob), and a potential eavesdropper (Eve), as illustrated in Fig. \ref{fig:SA_system1}. 
In contrast to the conventional cooperative RIS, the RIS in this paper is operated by an independent and market-oriented service provider whose objective is to offer reflection resources through market-based allocation.

To enable flexible and fine-grained resource allocation, the RIS equipped with $N$ reflecting elements is partitioned into $M$ disjoint modules, each comprising $L=N/M$ reflecting elements.
Each RIS module is treated as an indivisible economic resource and is allocated sequentially to competing bidders.
Once a module is allocated, the rights of ownership and control for that module are exclusively assigned to the winning bidder for the duration of the remaining transmission episode.
Bob and Eve are admitted bidders in the RIS resource market, and they compete for the temporary phase-control rights of each RIS resource module.
All nodes are located on a two-dimensional Cartesian plane with coordinates $\mathbf{x}_A, \mathbf{x}_B, \mathbf{x}_E$, and $\mathbf{x}_R \in \mathbb{R}^2$. 
The distance between node $i$ and node $j$ is denoted by $d_{ij}=\|\mathbf{x}_i-\mathbf{x}_j\|$, $i,j\in\{A,B,E,R\}$.
Since Bob and Eve are located at different positions, they experience different channel conditions and therefore obtain different marginal performance gains from additional RIS modules.
These differences motivate strategic competition for RIS reflection resource under budget constraints.

\begin{figure}[t]
	\centering
	\includegraphics[width=1.0\linewidth]{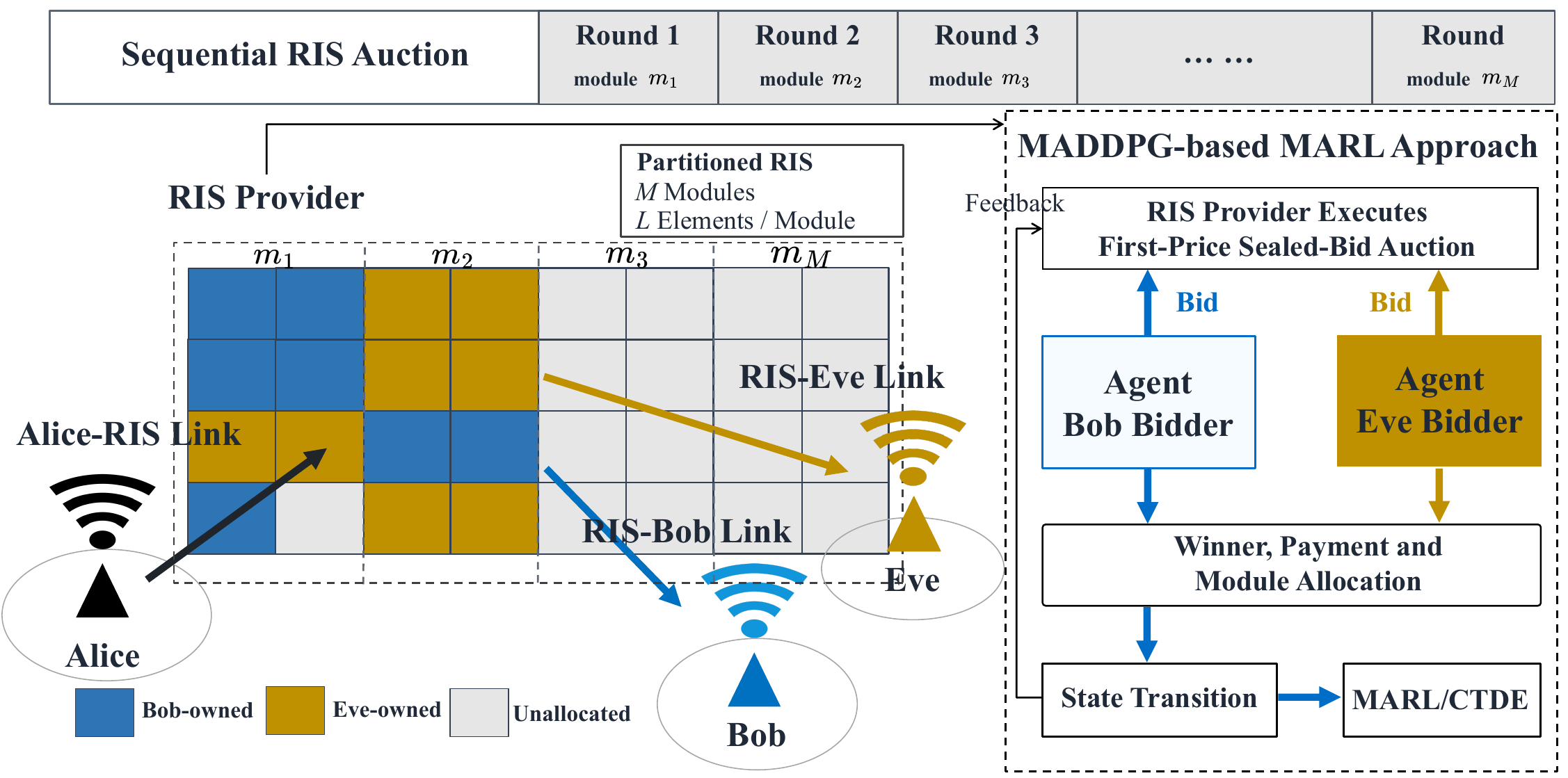}
	\caption{An auction-driven RIS-assisted network.}
	\label{fig:SA_system1}
	\vspace{-20pt}
\end{figure}

\subsection{Channel Model}\label{subsection_sys_ch}

We consider a narrowband quasi-static block-fading channel model, in which channel coefficients remain constant within one transmission episode and vary independently across episodes. 
Thus, the complex baseband channel between Alice and RIS is denoted by $\mathbf{h}_{AR}\in\mathbb{C}^{N\times 1}$ and is modeled as
\begin{equation}
	\mathbf{h}_{AR}=\sqrt{C_0 d_{AR}^{-\alpha}}\,\mathbf{g}_{AR},
\end{equation}
where $C_0$ denotes the path loss at a reference distance, $\alpha$ represents the path-loss exponent, and $\mathbf{g}_{AR}$ follows a Rician fading distribution to capture both line-of-sight (LoS) and non-line-of-sight (NLoS) components. 
Similarly, the channels from the RIS to Bob and Eve are given by
\begin{equation}
	\mathbf{h}_{RB}=\sqrt{C_0 d_{RB}^{-\alpha}}\,\mathbf{g}_{RB},\quad
	\mathbf{h}_{RE}=\sqrt{C_0 d_{RE}^{-\alpha}}\,\mathbf{g}_{RE}.
\end{equation}

Let $\mathcal{M}=\{1,2,\ldots,M\}$ denote the set of RIS modules.
For each module $m\in\mathcal{M}$, the phase-shift matrix is defined as
$\mathbf{\Phi}_m = \text{diag}(\boldsymbol{\phi}_m) \in \mathbb{C}^{L \times L}$, 
where  
$\boldsymbol{\phi}_m = [e^{j\theta_{m,1}}, \dots, e^{j\theta_{m,L}}]^T \in \mathbb{C}^{L \times 1}$  is the reflection coefficient vector, and  $\theta_{m,\ell} \in [0, 2\pi)$ is the phase shift of the $\ell$-th element in RIS module $m$.
Since the $M$ modules are mutually disjoint, the overall RIS phase-shift matrix can be constructed by stacking the module-level phase-shift matrices in a block-diagonal form, i.e., $\boldsymbol{\Phi} = \text{diag}(\boldsymbol{\phi}_1, \dots, \boldsymbol{\phi}_M) \in \mathbb{C}^{N \times N}$.   

Moreover, the direct links from Alice to Bob and Eve are assumed to be blocked or severely attenuated to facilitate the evaluation of the impact of RIS resource allocation on secrecy performance, which has been commonly adopted in dense urban environments.  
Let $\mathcal{S}_i\subseteq\mathcal{M}$ denote the set of modules acquired by bidder $i\in\{B,E\}$, the effective cascaded channel from Alice to bidder $i$ via the allocated RIS modules is expressed as
\begin{equation}\label{eq_hAI}
	h_{Ai}^{\text{RIS}} = \sum_{m \in \mathcal{S}_i} \mathbf{h}_{Ri,m}^H \mathbf{\Phi}_m \mathbf{h}_{AR,m},
\end{equation}
where $\mathbf{h}_{Ri,m} \in \mathbb{C}^{L \times 1}$ and $\mathbf{h}_{AR,m} \in \mathbb{C}^{L \times 1}$ denote the corresponding channels of RIS module $m$.
We assume a quasi-static channel state information (CSI) model at the episode level, where CSI is updated at the beginning of each episode and remains fixed during the auction rounds. Each bidder uses its own CSI to compute the private valuation, and neither the CSI nor the valuation is disclosed to the opponent.

\subsection{Physical-Layer Security Model}\label{subsection_sys_pls}

Based on the channel model, the received signals at bidder $i\in\{B,E\}$ is
\begin{equation}
	y_i=\sqrt{P}\,(h_{Ai}^{\mathrm{RIS}}+h_{Ai}^{\mathrm{dir}})\,x+n_i,  
\end{equation}
where $h_{Ai}^{\mathrm{RIS}}$ is the Alice-RIS-$i$ cascaded channel defined in (\ref{eq_hAI}), and $h_{Ai}^{\mathrm{dir}}$ is the Alice-$i$ direct channel. 
The term $P$ is the transmit power, $x$ is the information-bearing symbol with $\mathbb{E}[|x|^2]=1$, and $n_i\sim\mathcal{CN}(0,\sigma^2)$  is the AWGN at receiver $i$. The corresponding signal-to-noise ratios (SNR) are
\begin{equation}
	\gamma_i=\frac{P\,|h_{Ai}^{\mathrm{RIS}}+h_{Ai}^{\mathrm{dir}}|^2}{\sigma^2}.
\end{equation} 
Accordingly, the achievable communication rates are 
\begin{equation}\label{eq_Ri}
	R_i=\log_2(1+\gamma_i).
\end{equation}
The instantaneous secrecy rate is defined as $	R_s= [R_B -R_E  ]^+$, 
where $[\cdot]^+=\max(\cdot,0)$.

\subsection{Threat Model}\label{subsection_sys_threat}

Eve remains passive during wireless signal transmission but actively participates in the RIS resource market to improve its interception capability. 
Specifically, Eve does not inject active wireless interference, but it acquires RIS modules through bidding and configures the won modules to strengthen the Alice-Eve cascaded link. 
The RIS provider performs registration, billing, and contract enforcement for authorized bidders. 
It therefore knows each bidder's authenticated service identity and can prevent anonymous bidding and budget violations, but cannot reliably infer from a service request alone whether a registered bidder is an intended receiver or a potential eavesdropper. 
Similarly, Bob cannot perfectly identify malicious bidders during the sealed-bid auction and observes only limited auction feedback, such as previous winning prices and its own allocation outcomes. 
This information asymmetry creates a market-enabled eavesdropping risk, in which Eve acquires RIS control rights through authorized market participation before transmission and exploits them for passive interception during transmission.

\section{Sequential RIS Auction (SRA) Framework and Physical-Layer Security Analysis}\label{sec_SRA}

In this section, we develop the SRA framework by specifying the auction setup and RIS module valuation model, and further characterize the physical-layer security implications of RIS module allocation.

\subsection{Auction Setup }\label{subsec:auction_setup}

We consider a market-oriented RIS deployment, where a dedicated RIS provider supplies the $M$ RIS modules and allocates them to competing bidders through a sequential first-price sealed-bid auction \cite{Sun2006}.
Under this framework, the bidder that acquires module $m$ obtains the temporary phase-control right of this module. This setting reflects practical RIS scenarios, in which reflection resources are economically priced and competitively allocated. Bob aims to enhance the legitimate link, whereas Eve seeks to strengthen its interception capability.
At auction round $t\in\{1,\ldots,T_a\}$, where $T_a=M$ is the number of auction rounds with one module allocated per round, the provider offers one module $m$ to two strategic bidders: Bob and Eve.
Each bidder $i\in\{B,E\}$ submits a bid $b_{i,m_t}(t)$ for the current module. After allocation, the auction proceeds to the next round with updated budgets and allocation sets.

Since each module can be allocated to at most one bidder, the allocation sets satisfy
$\mathcal{S}_B(t)\cap\mathcal{S}_E(t)=\emptyset$. 
The allocation sets evolve according to
\begin{equation}
	\mathcal{S}_i(t+1)=
	\begin{cases}
		\mathcal{S}_i(t)\cup\{m\}, & \mathbb{I}_i(t)=1,\\
		\mathcal{S}_i(t), & \mathbb{I}_i(t)=0.
	\end{cases}, \forall t,
\end{equation}
where $\mathbb{I}_i(t)=1$ if bidder $i$ wins module $m$ at round $t$, and $0$ otherwise. After all auction rounds are completed, the final allocation set is given by $\mathcal{S}_i=\mathcal{S}_i(T_a+1)$.

\subsection{Module Valuation and Bidding Rule}\label{subsec:valuation_model} 

To connect physical-layer performance with economic decision-making, we introduce a conversion factor $\lambda>0$, which maps the communication rate improvement into monetary value. Based on~\eqref{eq_Ri}, the valuation of bidder $i$ for the module $m$ is defined as
\begin{equation}
	v_{i,m}(t)
	=\lambda\!\left[R_i\big(\mathcal{S}_i(t)\cup\{m\}\big)-R_i\big(\mathcal{S}_i(t)\big)\right],
	\label{eq:valuation}
\end{equation}
which represents the marginal economic value of acquiring module $m$.
For both Bob and Eve, $R_i(\cdot)$ denotes the achievable communication rate under the current allocation. 
The noise power is treated as a fixed system parameter during the auction-stage optimization, so the private valuation depends only on the bidder's own CSI and allocation state.
Thus, the valuation captures the rate gain enabled by acquiring the RIS module.
Since the marginal rate gain depends on the modules obtained in previous rounds, the valuation is state-dependent and evolves throughout the auction process.

At the beginning of round $t$, bidder $i$ has a remaining budget $B_i(t)$. 
For Bob, the budget represents the payment capacity allocated to enhancing the legitimate link, whereas for Eve, it represents the payment capacity allocated to improving its interception capability.
The budget constraint prevents either bidder from acquiring RIS modules without considering the impact of current payments on future auction opportunities.

Bidders are not assumed to follow truthful bidding behavior. 
Instead, each bidder adopts a bidding strategy  that depends on the valuation and the remaining budget. 
Specifically, the bid of bidder $i$ is modeled as
\begin{equation}\label{eq:bid}
	b_{i,m}(t)=\min\!\left\{a_i(t)\,v_{i,m}(t),\,B_i(t)\right\},
\end{equation}
where $a_i(t)\in[0,1]$ denotes a  bidding strategy and $B_i(t)$ denotes the remaining budget at the beginning of round $t$. 
This parametric form captures the bid-shading behavior inherent in first-price auctions while maintaining a low-dimensional and tractable action space for learning.

The module is allocated to the highest bidder, and the winner pays the submitted bid according to the first-price payment rule. 
The budget of each bidder evolves as
\begin{equation}
	B_i(t+1)=B_i(t)-\mathbb{I}_i(t)\, b_{i,m}(t).
\end{equation}
This payment mechanism induces a trade-off between winning probability and payment cost, leading to forward-looking strategic behavior under budget constraints.

\subsection{Bidding Problem Formulation}\label{subsec:problem_formulation}

Because the secrecy rate depends not only on the module currently allocated but also on the cumulative allocation set $\mathcal{S}_i(t)$, the bidding decisions of Bob and Eve are temporally coupled and must be optimized jointly across all $T_a$ rounds.
We formulate bidder $i\in\{B,E\}$'s bidding problem as a finite-horizon sequential decision problem under budget constraints:

\begin{subequations}\label{eq:P1} 
\begin{align}
&  \max_{\{a_i(t)\}_{t=1}^{T_a}} \quad
U_i = \mathbb{E}\left[\sum_{t=1}^{T_a}\delta^{t-1}\,r_i(t) \right] \\
\textnormal{s.t.}\quad 
& a_i(t)\in[0,1], \quad \forall t\in\{1,\dots,T_a\}, \\
& B_i(t+1)=B_i(t)-\mathbb{I}_i(t)\,b_{i,m}(t),\ B_i(t)\ge 0,  \\
& \mathcal{S}_B(t)\cap\mathcal{S}_E(t)=\emptyset,\ \forall t, \\
& b_{i,m}(t)=\min\{a_i(t)\,v_{i,m}(t),\,B_i(t)\}.  
\end{align} 
\end{subequations}
where $r_i(t)=\mathbb{I}_i(t)\big[v_{i,m}(t)-b_{i,m}(t)\big]$ is the instantaneous economic surplus and $\delta\in(0,1)$ is the discount factor.
Note that 
problem~(\ref{eq:P1}) is intractable in closed form for three reasons. 
(i)~The valuation $v_{i,m}(t)$ is state-dependent on the historical allocation $\mathcal{S}_i(t)$ and on the opponent's unknown strategy, so the objective $U_i$ has no recursive closed form. 
(ii)~Bob and Eve have  incomplete information on each other's CSI, valuation, and remaining budget, which prevents the application of standard Bayesian equilibrium analysis. 
(iii)~The two bidders act  simultaneously, turning (\ref{eq:P1}) into a two-player stochastic game rather than a single-bidder MDP.
Moreover, the continuous bidding ratio $a_i(t)$ and the combinatorial allocation state cause the state-action space to grow exponentially with $M$, ruling out value iteration and exact dynamic programming. 
These properties motivate the multi-bidder reinforcement learning framework developed in Section~\ref{sec_mecha}.

\subsection{Security-Economic Coupling in RIS Module Allocation}
\label{subsec:physical_layer_analysis}

To connect economic auction decisions with physical-layer secrecy performance, we characterize how each module allocation changes the secrecy rate and derive the corresponding valuation-driven winning condition.

\begin{theorem}[Valuation-Driven Winning Condition] \label{theorem}
	Consider auction round $t$ for module $m$, and define the marginal rate gain of bidder $i\in\{B,E\}$ as
	\begin{align}
		\Delta_{i,m}(t)
		&=
		R_i\bigl(\mathcal{S}_i(t)\cup\{m\}\bigr)
		-
		R_i\bigl(\mathcal{S}_i(t)\bigr).
		\label{eq:marginal_rate_gain}
	\end{align}
	Under the module-isolated channel, we have
	$\Delta_{i,m}(t)\geq 0$.
	The secrecy-rate variation satisfies
	\begin{align}
	&	0
		\leq R_s(t+1)-R_s(t)
		\leq \Delta_{B,m}(t),
		&& \text{if Bob wins},
		\label{eq:secrecy_change_bob}\\
	&	-\Delta_{E,m}(t)
		\leq R_s(t+1)-R_s(t)
		\leq 0,
		&& \text{if Eve wins}.
		\label{eq:secrecy_change_eve}
	\end{align}
	
Excluding tie cases,  Bob wins module $m$ if and only if
	\begin{align}
	&	\min\!\left\{
		a_B(t)\lambda\Delta_{B,m}(t),\,B_B(t)
		\right\}\nonumber\\
	&	> 
		\min\!\left\{
		a_E(t)\lambda\Delta_{E,m}(t),\,B_E(t)
		\right\}.
		\label{eq:bob_winning_condition}
	\end{align}
	If neither bidder is budget-saturated, i.e.,
	$a_i(t)\lambda\Delta_{i,m}(t)\leq B_i(t)$ for $i\in\{B,E\}$,
	and $a_E(t)\Delta_{B,m}(t)>0$, the winning condition reduces to
	\begin{align}
		\frac{a_B(t)}{a_E(t)}
		>
		\frac{\Delta_{E,m}(t)}{\Delta_{B,m}(t)}.
		\label{eq:bidding_ratio_threshold}
	\end{align}
\end{theorem}

\begin{IEEEproof}
	Under per-module phase alignment, assigning module $m$ to bidder $i$ cannot decrease its achievable rate, i.e., 
	$\Delta_{i,m}(t)\geq 0$.
	Let $x(t)=R_B(t)-R_E(t)$. If Bob wins module $m$, then
	\begin{align}
		R_s(t+1)-R_s(t)
		&=
		\bigl[x(t)+\Delta_{B,m}(t)\bigr]^+
		-
		\bigl[x(t)\bigr]^+.
		\label{eq:proof_bob}
	\end{align}
	Since $[x]^+$ is nondecreasing and 1-Lipschitz, the difference in
	\eqref{eq:proof_bob} lies in
	$[0,\Delta_{B,m}(t)]$, which proves \eqref{eq:secrecy_change_bob}.
	If Eve wins module $m$, then
	\begin{align}
		R_s(t+1)-R_s(t)
		&=
		\bigl[x(t)-\Delta_{E,m}(t)\bigr]^+
		-
		\bigl[x(t)\bigr]^+,
		\label{eq:proof_eve}
	\end{align}
	which lies in
	$[-\Delta_{E,m}(t),0]$, proving \eqref{eq:secrecy_change_eve}.
	Finally, \eqref{eq:bob_winning_condition} follows directly from the first-price sealed-bid rule and the bid definition in \eqref{eq:bid}. When both budget constraints are inactive, substituting
	$v_{i,m}(t)=\lambda\Delta_{i,m}(t)$ and cancelling the common factor $\lambda$ yields \eqref{eq:bidding_ratio_threshold}.
\end{IEEEproof}

 \begin{figure}[t]
 	\centering
 	\includegraphics[width=\linewidth]{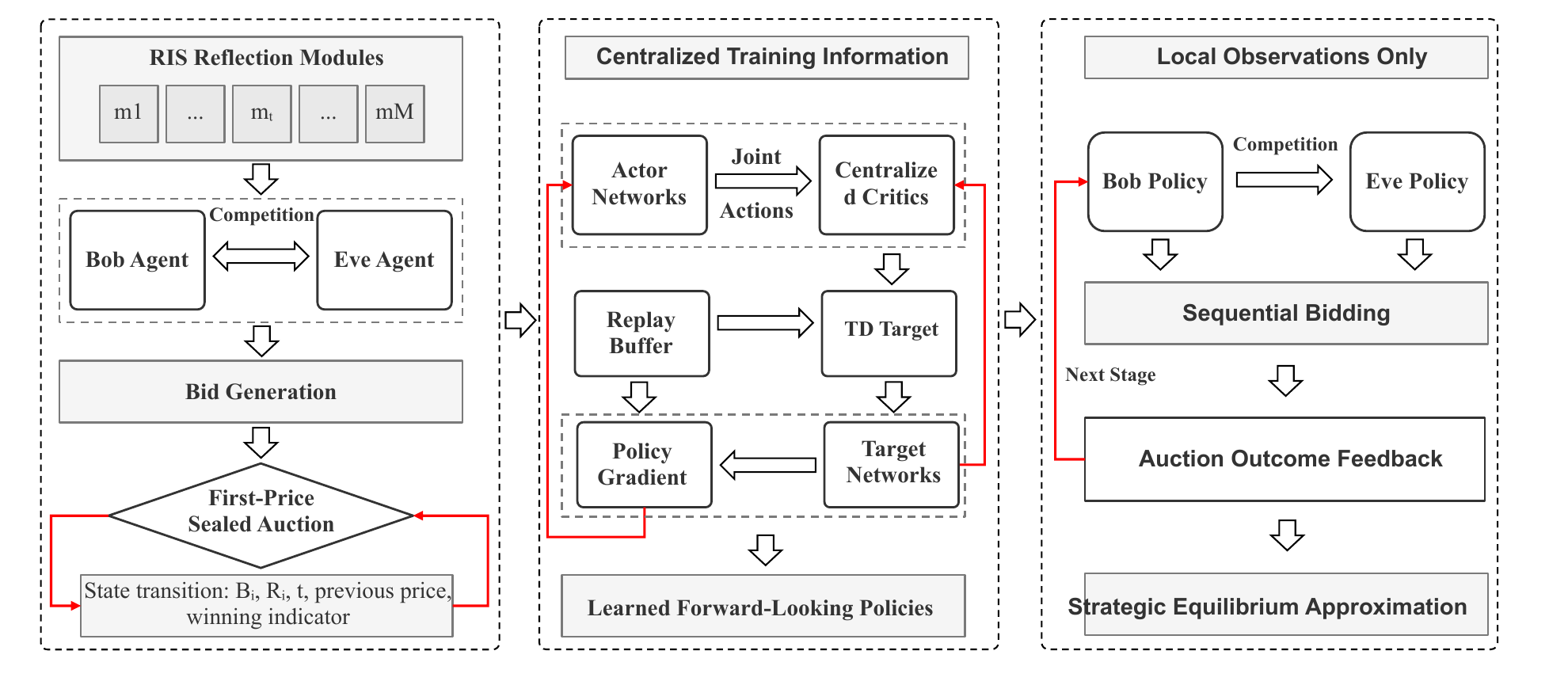}
 	\caption{SRA Mechanism with MARL-Based Bidding}
 	\label{fig:fig2_frame}
 	\vspace{-10pt}
 \end{figure}
 
Theorem~\ref{theorem} reveals the security-economic coupling induced by RIS module allocation. Although the bidding policy is trained with economic surplus rewards, its effect on secrecy is induced through module ownership and the resulting rate increments $\Delta_{B,m}(t)$ and $\Delta_{E,m}(t)$. This coupling also provides the interface to the learning-based solution, where auction states, bidding actions, surplus rewards, and state transitions are defined according to the SRA mechanism.

\section{MARL-Based Bidding Strategy}\label{sec_mecha}

Given problem~\eqref{eq:P1}, this section develops a learning-based solution for the SRA mechanism, as shown in \ref{fig:fig2_frame}. 
We first instantiate the sequential bidding problem as a Markov game by defining the state, action, reward, and state transition induced by module allocation and budget updates. We then adopt a CTDE-based MADDPG framework, where each bidder learns a forward-looking policy that maps local auction observations to budget-feasible bidding ratios and maximizes long-term economic surplus.

\subsection{Markov Game Instantiation}

In the SRA mechanism, RIS modules are allocated sequentially under budget constraints. 
The allocation outcome at each auction round affects subsequent rounds through the accumulated allocated modules and the remaining budget
\cite{Ahmad2023, Dai2026}. Therefore, we model the SRA as a two-bidder stochastic Markov game.
Let $\mathcal{N}=\{B,E\}$ denote the bidder set. The Markov game is represented by the state-action-reward-transition tuple $\mathcal{G}_{\mathrm{SRA}}=\langle\mathcal{S},\mathcal{A},\mathcal{R},\mathcal{P}\rangle$, where $\mathcal{S}$ denotes the auction state space, $\mathcal{A}=\mathcal{A}_B\times\mathcal{A}_E$ denotes the joint action space, $\mathcal{R}=\{r_B,r_E\}$ denotes the reward functions, and $\mathcal{P}$ denotes the transition kernel induced by the first-price auction rule and the budget update.
The key components of the Markov game are defined as follows.

\textbf{a) State:}
At auction round $t$, centralized training can use the full auction state defined in Section~\ref{sec_SRA}, which gathers the allocation history $\mathcal{S}_B(t)$, $\mathcal{S}_E(t)$, the budgets $B_B(t)$, $B_E(t)$, the instantaneous rates $R_B(t)$, $R_E(t)$, the secrecy rate $R_s(t)$, and the round index $t$. Following the standard centralized-training decentralized-execution (CTDE), bidder $i$ instead observes a partial, privacy-preserving local state vector
\begin{align}
	&\mathbf{s}_i(t)= \nonumber \\
	&\Big[
	R_i(t-1),
	\frac{B_i(t)}{B_i(0)},
	\frac{t}{T_a},\;
	v_{i,m}(t),\;
	\frac{p_{t-1}}{B_i(0)},
	\mathbb{I}(w_{t-1}=i)
	\Big],
	\label{eq:state_space}
\end{align}
which includes the accumulated communication rate, normalized remaining budget, auction progress, current valuation, previous clearing price, and recent winning outcome.

This state representation captures the principal factors governing sequential auction dynamics.
In particular, the remaining budget defines the feasible bidding space, while the valuation reflects the marginal communication gain of the current module.
The accumulated rate encodes historical allocation outcomes, and the auction progress provides temporal context.
Moreover, the previous price and winning indicator convey recent competition intensity and opponent behavior.
Overall, the proposed observation integrates both economic constraints and competitive signals, enabling bidders to infer opponent strategies and adapt their bidding strategies accordingly.

\textbf{b) Action:}
At each round, bidder $i$ selects a continuous bidding ratio
\begin{equation}
	a_i(t)\in[0,1],
\end{equation}
which controls the aggressiveness of its bid relative to the current module valuation. 
The actual bid $b_{i,m}(t)$ is then generated according to the economic and physical constraints defined in~\eqref{eq:bid}. 
By optimizing a normalized ratio rather than an absolute monetary value, the reinforcement learning bidder operates within a bounded and well-structured action space. 
This formulation guarantees feasibility under budget constraints while enabling the policy to learn complex, non-truthful bidding strategies that optimize long-term utility.

\textbf{c) Reward:}
The instantaneous reward is defined as
\begin{equation}\label{eq:reward}
	r_i(t)\!=\!
	\begin{cases}
		v_{i,m}(t)-b_{i,m}(t), &\!\! \text{if bidder $i$ wins module $m$},\!\\
		0, &\!\! \text{otherwise}.
	\end{cases}
\end{equation}
This reward corresponds to the economic surplus achieved at each auction round and captures the trade-off between aggressive bidding and budget preservation. Although secrecy is not explicitly encoded in the reward, the reward is physically grounded through the valuation $v_{i,m}(t)=\lambda\Delta_{i,m}(t)$. For Bob, a larger valuation corresponds to a larger legitimate-link rate gain $\Delta_{B,m}(t)$ and thus promotes the acquisition of modules that increase $R_B$; for Eve, a larger valuation corresponds to a larger interception-link rate gain $\Delta_{E,m}(t)$ and thus promotes the acquisition of modules that increase $R_E$. As established in Theorem~\ref{theorem}, the winner of each module determines whether the subsequent secrecy-rate variation is governed by $\Delta_{B,m}(t)$ or $\Delta_{E,m}(t)$. 
Therefore, the secrecy performance emerges from the competition between surplus-driven legitimate-link enhancement and surplus-driven interception-link enhancement under budget constraints.


\textbf{d) State Transition:}
Given the current module $m_t$ and the joint actions $\{a_B(t),a_E(t)\}$, the bids are computed according to~\eqref{eq:bid}. Ignoring tie cases, the winner and clearing price are given by
\begin{align}
&	w_t = \arg\max_{i\in\{B,E\}} b_{i,m_t}(t), \\
 &	p_t= b_{w_t,m_t}(t), \\
&	\mathbb{I}_i(t) = \mathbb{I}(w_t=i).
\end{align}
The auction state then evolves as
\begin{align}
	B_i(t+1) &= B_i(t)-\mathbb{I}_i(t)p_t, \\
	\mathcal{S}_i(t+1) &=
	\begin{cases}
		\mathcal{S}_i(t)\cup\{m_t\}, & \mathbb{I}_i(t)=1,\\
		\mathcal{S}_i(t), & \mathbb{I}_i(t)=0,
	\end{cases} \\
	R_i(t+1) &= R_i\big(\mathcal{S}_i(t+1)\big), \\
	R_s(t+1) &= [R_B(t+1)-R_E(t+1)]^+ .
\end{align}
Since the channel realization is fixed within an episode and the related parameters are fully determined by the current state and joint actions, the transition satisfies the Markov property.

\subsection{MADDPG under CTDE}

In the SRA mechanism, each bidder chooses a continuous bidding ratio, which makes multi-bidder deep deterministic policy gradient (MADDPG) suitable for learning deterministic continuous-control policies. Meanwhile, Bob and Eve update their bidding strategies simultaneously, making the auction environment non-stationary from the perspective of each individual bidder. The CTDE addresses this issue by using global auction information during training while allowing each bidder to execute its policy based only on local observations. 
Therefore, in this section, we adopt the MADDPG framework to solve the Markov game $\mathcal{G}_{\mathrm{SRA}}$ under the centralized training and decentralized execution (CTDE).
The centralized critic is used only during offline training, where the training environment has access to the global auction state and joint actions. During decentralized execution, bidder $i$ selects $a_i(t)$ solely from its local observation $\mathbf{s}_i(t)$, and its CSI and private valuation are not disclosed to the opponent.

\subsubsection{Policy Representation}

Each bidder $i \in \{B,E\}$ is equipped with a deterministic policy network $\mu_i(\mathbf{s}_i; \phi_i)$ that maps its local observation $\mathbf{s}_i$ to a continuous bidding ratio $a_i(t) \in [0,1]$.  
The monetary bid is subsequently determined by the valuation and remaining budget according to~\eqref{eq:bid}. Therefore, the actor only needs to learn a normalized bidding policy, while bid feasibility is enforced by the SRA mechanism.

\subsubsection{Centralized Critic and Value Estimation}

During training, each bidder maintains a centralized critic $Q_i(\mathbf{s}, a_B, a_E; \theta_i)$ that evaluates joint actions under the global state $\mathbf{s} = [\mathbf{s}_B, \mathbf{s}_E]$. 
To mitigate the non-stationarity induced by simultaneous policy updates, the critic is trained by minimizing the temporal-difference (TD) loss
\begin{equation}
	L(\theta_i) = \mathbb{E} \left[ \left( y_i - Q_i(\mathbf{s}, a_B, a_E) \right)^2 \right].
\end{equation}
The target value is computed using target networks as
\begin{equation}
	y_i = r_i + \delta Q_i'(\mathbf{s}', a_B', a_E') \big|_{a_j' = \mu_j'(\mathbf{s}_j')}.
\end{equation}

Here, $r_i$ denotes the economic surplus reward defined in~\eqref{eq:reward}, and $\delta$ is the discount factor. 
Leveraging the global state and joint actions during training yields a more stable value estimate that captures the evolving strategies of competing bidders, while execution remains fully decentralized and relies only on local observations.

\subsubsection{Policy Update}

The actor parameters $\phi_i$ are updated via gradient ascent on the expected return. 
According to the deterministic policy gradient theorem, the gradient is given by
\begin{equation}
	\nabla_{\phi_i} J(\phi_i) \!=\! \mathbb{E} \left[ \nabla_{a_i} Q_i(\mathbf{s}, a_B, a_E) \big|_{a_i = \mu_i(\mathbf{s}_i)} \nabla_{\phi_i} \mu_i(\mathbf{s}_i) \right].\!
\end{equation} 
This formulation establishes a direct link between local policy parameters and the global value function, enabling each bidder to adjust its bidding ratio toward maximizing long-term utility. 
Through feedback from the centralized critic, the learned policy internalizes inter-temporal budget coupling and anticipates opponent responses, gradually approaching a strategic equilibrium over the auction horizon.

\subsection{Algorithm and Discussion}

From a game-theoretic perspective, the learned policies can be interpreted as approximate best responses under repeated competition. 
Compared with static or myopic bidding rules, the MARL-based approach enables bidders to anticipate future opportunities and adapt dynamically to environmental and strategic changes. 
Therefore, the proposed framework provides a computational approach for approximating equilibrium behavior in dynamic competitive environments with budget constraints.

For clarity, the detailed update rules follow the standard MADDPG framework and are omitted in Algorithm~\ref{alg:SA}. 
The overall training and interaction procedure is summarized as follows.

\begin{algorithm}[t]
	
	\footnotesize
	\caption{MARL-Based Sequential Auction under CTDE Framework}
	\label{alg:SA} 
	\begin{algorithmic}[1]
		\Require
		Number of RIS modules $M$ (auction rounds $T_a=M$); initial budgets $\{B_i(0)\}$; discount factor $\delta$
		\Ensure 
		Learned policies $\{\mu_i\}$
		
		\State Initialize actor $\mu_i$, critic $Q_i$, and target networks
		\State Initialize replay buffer $\mathcal{D}$
		
		\For{each episode}
		\State Initialize state $\mathbf{s}(1)$
		
		\For{$t = 1$ to $T_a$}
		\For{each bidder $i$}
		\State Observe $\mathbf{s}_i(t)$ and select $a_i(t) = \mu_i(\mathbf{s}_i(t)) + \epsilon$
		\State Compute bid $b_i(t)$
		\EndFor
		
		\State Execute auction and update budgets
		\State Observe $(\mathbf{r}(t), \mathbf{s}(t+1))$
		\State Store transition in $\mathcal{D}$
		\EndFor
		
		\State Update critic via TD loss
		\State Update actor via policy gradient
		\State Soft update target networks
		\EndFor
	\end{algorithmic}
\end{algorithm}

\section{Simulation Results and Discussion}\label{sec_simul}

In this section, we evaluate the proposed SRA framework by examining how sequential allocation, budget coupling, and strategic bidding affect secrecy performance, bidding behavior, and RIS-provider revenue. We further assess the security and economic efficiency of the resulting auction outcomes.

\begin{table}[t]
	\centering
	\caption{Simulation Parameters}
	\vspace{-7pt}
	\label{tab:sim_params}
	\begin{tabular}{lp{0.45\linewidth}}
		\hline
		\textbf{Parameter} & \textbf{Value} \\
		\hline
		\multicolumn{2}{c}{\textit{Physical Layer and Channel Settings}} \\
		RIS size & $N=64$ ($8 \times 8$) elements  \\
		Number of RIS modules & $M=8$ \\
		Elements per RIS module & $L=N/M=8$ elements \\ 
		Channel model & Rician fading ($K=10$ dB) \\
		Pathloss exponent & $\alpha = 2.2$ \\
		Transmit power & $40$~dBm \\
		Noise power & $-90$~dBm \\
		\hline
		\multicolumn{2}{c}{\textit{Sequential Auction Parameters}} \\
		Initial budget (Bob/Eve) & $B_B = B_E = 10$ \\
		Auction rounds & $T_a=8$ (one module per round) \\
		Auction rule & First-price sealed-bid \\
		Auction termination & All modules allocated or budget exhausted \\
		\hline
		\multicolumn{2}{c}{\textit{MADDPG Training Parameters}} \\
		Actor/Critic hidden layers & 2 layers, 256 neurons each \\
		Actor learning rate & $1\times10^{-4}$ \\
		Critic learning rate & $3\times10^{-4}$ \\
		Discount factor & $\delta = 0.98$ \\
		Target network update rate & $\tau = 0.005$ \\
		Replay buffer size & $10^5$ \\
		Training episodes & $10{,}000$ \\
		Exploration noise (start/end) & $0.5 \rightarrow 0.001$ \\
		\hline
	\end{tabular}
	\vspace{-10pt}
\end{table}

\subsection{Simulation Setup and Baseline Strategies} 

Table~\ref{tab:sim_params} summarizes the main simulation parameters. 
For a fair comparison, all bidding strategies are evaluated using the same network topology, RIS partition, channel realizations, initial budgets, and module-valuation model. 
Unless otherwise specified, each reported result is averaged over $N_{\mathrm{eval}}=100$ independent evaluation episodes after training convergence, with the learned policies fixed during evaluation.
An $8\times8$ RIS comprising $N_{\mathrm{RIS}}=64$ reflecting elements is partitioned into $M=8$ disjoint modules, each containing $L=8$ elements. In each of the $T_a=M$ auction rounds, one module is offered, and Bob and Eve submit bids according to the first-price sealed-bid rule. All wireless links follow Rician fading to capture both line-of-sight and non-line-of-sight components. For every bidding strategy, the valuation of the current module is computed from its state-dependent marginal rate gain according to Section~\ref{subsec:valuation_model}.
All monetary quantities are normalized, and $\lambda=1$ is used to map one unit of marginal rate gain into one normalized budget unit.

Under the SRA mechanism, bidders decide sequentially from their available auction information (remaining budget, current valuation, and public feedback on previous outcomes). 
To provide benchmarks for evaluating learning-based bidding strategies, we consider two representative baseline strategies that capture limited rationality and myopic decision-making behavior.

\textbf{1) Fixed Bidding Strategy:} 
In the fixed bidding strategy, each bidder uses a constant valuation-based bidding ratio throughout the auction process. 
Specifically, at each round $t$, the bid is determined as a fixed fraction of the current module valuation, i.e.,
$b_{i,m}(t) = \min\{\bar{a}_i v_{i,m}(t), B_i(t)\}$, where $\bar{a}_i \in [0, 1]$ is a predetermined constant.   

\textbf{2) Random Bidding Strategy:} 
In the random bidding strategy, each bidder independently decides whether to participate in the current auction round with a predefined probability, subject to the remaining budget. 
If participation occurs, the bid is randomly drawn from the feasible range determined by the available budget. 
This strategy does not exploit module valuation, past auction outcomes, or future opportunities, and thus represents an uninformed stochastic bidding behavior.

\begin{figure}[t]
	\centering
	\includegraphics[width=0.7\columnwidth]{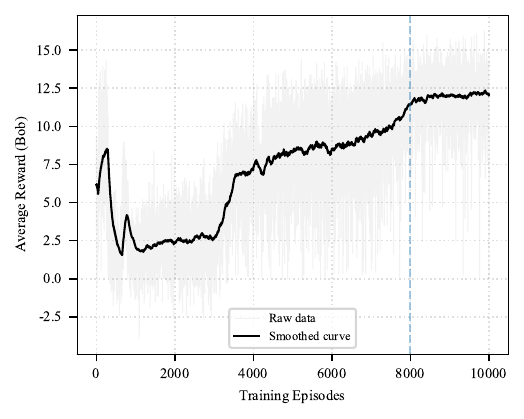}
	\caption{Training process of the RL-based bidding strategy}
	\label{fig:training_curve}
	\vspace{-10pt}
\end{figure}

Fig.~\ref{fig:training_curve} illustrates the training trajectory of the proposed RL-based bidding strategy. The raw reward exhibits substantial fluctuations owing to the stochastic nature of wireless channels and the competitive multi-bidder environment. Consequently, a moving average is utilized to reveal the underlying learning trend.
The smoothed reward increases from approximately 2.5 to 12.0 and then stabilizes, indicating that the MADDPG-based method reaches a stable bidding strategy despite channel randomness and simultaneous policy updates.

\vspace{-7pt}
\subsection{Secrecy-Rate Performance Against Different Eve Strategies}
\begin{table}[t]
	\centering
	\caption{Row-average performance over all Eve strategies. }
	\vspace{-7pt}
	\label{tab:SA_summary}
	\setlength{\tabcolsep}{4pt}
	\renewcommand{\arraystretch}{1.1}
	\footnotesize
	\begin{tabular}{lccccc}
		\toprule
		Bob strategy & Reward $\uparrow$ & Secrecy $\uparrow$ & RIS Rev. & Bid & Sec./Bid $\uparrow$ \\
		\midrule
		RL     & \textbf{\underline{10.34}} 
		& \textbf{\underline{2.70}} 
		& 3.89  
		& 3.55 
		& \textbf{\underline{1.03}} \\
		Random & 4.01 & 1.64 & 8.38 & 6.17 & 0.25 \\
		\midrule
		Fixed(0.1) & 2.49 & 0.32 & 5.74  & \textbf{\underline{0.21}} & 0.39 \\
		Fixed(0.2) & 4.39 & 0.96 & 5.35  & 1.04 & 0.44 \\
		Fixed(0.3) & 5.61 & 1.41 & 5.83  & 2.43 & \textbf{\underline{0.45}} \\
		Fixed(0.4) & \textbf{\underline{6.36}} & 2.00 & 6.57  & 4.15 & 0.41 \\
		Fixed(0.5) & 6.07 & 2.25 & 7.64  & 5.90 & 0.35 \\
		Fixed(0.6) & 5.30 & 2.42 & 8.87  & 7.90 & 0.31 \\
		Fixed(0.7) & 4.26 & \textbf{\underline{2.57}} & 9.93  & 9.38 & 0.27 \\
		Fixed(0.8) & 2.94 & 2.27 & 10.51 & 9.76 & 0.23 \\
		Fixed(0.9) & 1.73 & 1.89 & 10.87 & 9.86 & 0.19 \\
		Fixed(1.0) & 0.72 & 1.57 & \textbf{\underline{11.22}} & 9.91 & 0.16 \\
		\midrule
		Best fixed$^\dagger$ & 6.36 & 2.57 & 11.22 & 0.21 & 0.45 \\
		RL margin$^\dagger$  & +62.6\% & +5.1\% & -- & -- & +128.9\% \\
		\bottomrule
	\end{tabular} 
	\vspace{-15pt}
\end{table}

\begin{table*}[t]
	\centering
	\caption{Secrecy rate comparison with different strategies.}
	\vspace{-7pt}
	\label{tab:SA_secrecy}
	\setlength{\tabcolsep}{3pt}
	\renewcommand{\arraystretch}{1.05}
	\resizebox{\textwidth}{!}{
		\begin{tabular}{lccccccccccccc}
			\toprule
			Bob\textbackslash Eve & RL & Random & Fixed(0.1) & Fixed(0.2) & Fixed(0.3) & Fixed(0.4) & Fixed(0.5) & Fixed(0.6) & Fixed(0.7) & Fixed(0.8) & Fixed(0.9) & Fixed(1.0) & Avg. \\
			\midrule
			RL          
			& \textbf{\underline{2.50}}
			& \textbf{\underline{2.59}}
			& \textbf{\underline{2.76}}
			& \textbf{\underline{2.85}}
			& \textbf{\underline{2.49}}
			& \textbf{\underline{2.80}}
			& \textbf{\underline{2.83}}
			& \textbf{\underline{2.70}}
			& \textbf{\underline{2.80}}
			& \textbf{\underline{2.75}}
			& \textbf{\underline{2.71}}
			& \textbf{\underline{2.58}}
			& \textbf{\underline{2.70}} \\
			Random      & 1.62  & 1.80   & 2.72    & 2.54    & 2.28    & 2.10    & 1.81    & 1.54    & 1.25    & 0.98    & 0.78    & 0.58    & 1.64 \\
			\midrule
			Fixed(0.1)  & 0.00  & 0.05   & 2.75    & 0.98    & 0.09    & 0.01    & 0.00    & 0.00    & 0.00    & 0.00    & 0.00    & 0.00    & 0.32 \\
			Fixed(0.2)  & 0.00  & 0.51   & 2.86    & 2.75    & 1.95    & 0.87    & 0.37    & 0.11    & 0.04    & 0.00    & 0.00    & 0.00    & 0.96 \\
			Fixed(0.3)  & 0.09  & 1.36   & 2.84    & 2.83    & 2.75    & 2.31    & 1.60    & 0.93    & 0.52    & 0.23    & 0.11    & 0.06    & 1.41 \\
			Fixed(0.4)  & 0.38  & 2.06   & 2.85    & 2.84    & 2.84    & 2.72    & 2.43    & 1.92    & 1.36    & 0.94    & 0.60    & 0.38    & 2.00 \\
			Fixed(0.5)  & 1.12  & 2.49   & 2.84    & 2.85    & 2.85    & 2.82    & 2.77    & 2.54    & 2.20    & 1.67    & 1.27    & 0.94    & 2.25 \\
			Fixed(0.6)  & 1.70  & 2.68   & 2.86    & 2.85    & 2.86    & 2.83    & 2.81    & 2.74    & 2.54    & 2.26    & 1.94    & 1.57    & 2.42 \\
			Fixed(0.7)  & 2.43  & 2.66   & 2.80    & 2.78    & 2.76    & 2.75    & 2.71    & 2.67    & 2.58    & 2.43    & 2.23    & 1.92    & 2.57 \\
			Fixed(0.8)  & 2.37  & 2.30   & 2.42    & 2.39    & 2.37    & 2.33    & 2.30    & 2.25    & 2.20    & 2.15    & 2.07    & 1.91    & 2.27 \\
			Fixed(0.9)  & 1.98  & 1.92   & 2.01    & 1.99    & 1.95    & 1.93    & 1.91    & 1.87    & 1.82    & 1.79    & 1.78    & 1.67    & 1.89 \\
			Fixed(1.0)  & 1.67  & 1.60   & 1.68    & 1.62    & 1.61    & 1.59    & 1.58    & 1.52    & 1.51    & 1.47    & 1.43    & 1.41    & 1.57 \\
			\midrule
			Column Avg.
			& \textbf{\underline{1.32}}
			& \textbf{\underline{1.84}}
			& \textbf{\underline{2.62}}
			& \textbf{\underline{2.44}}
			& \textbf{\underline{2.23}}
			& \textbf{\underline{2.09}}
			& \textbf{\underline{1.93}}
			& \textbf{\underline{1.73}}
			& \textbf{\underline{1.57}}
			& \textbf{\underline{1.39}}
			& \textbf{\underline{1.24}}
			& \textbf{\underline{1.09}}
			& --- \\
			\bottomrule
	\end{tabular}}
	\vspace{-10pt}
\end{table*}

After training convergence, the learned bidding strategy is evaluated under the SRA mechanism against random and fixed baselines under heterogeneous Eve strategies. Table~\ref{tab:SA_summary} reports the row-average performance, where Bob-RL achieves the highest reward and secrecy rate. This indicates that the learned policy can jointly exploit state-dependent module valuations, payment costs, and remaining budgets, whereas fixed and random strategies are more sensitive to strategy mismatch.
Table~\ref{tab:SA_secrecy} further presents the complete Bob--Eve secrecy-rate matrix. Bob-RL achieves the highest row-average secrecy rate of $2.70$~bps/Hz and remains the best response to adaptive Eve-RL with $2.50$~bps/Hz. From the attacker side, Eve-RL reduces Bob's average secrecy rate to $1.32$~bps/Hz, revealing a strong adaptive threat and confirming the value of forward-looking bidding under strategic competition.

\vspace{-7pt}

\subsection{Secrecy Distribution and Physical-Layer Upper-Bound Verification}
\label{subsec:SA_performance}

\begin{table*}[t]
	\centering 
	\caption{Secrecy rate per unit cost under different bidding strategies.}
	\vspace{-7pt}
	\label{tab:SA_rate_bid_ratio}
	\setlength{\tabcolsep}{3pt}
	\renewcommand{\arraystretch}{1.05}
	\resizebox{\textwidth}{!}{
		\begin{tabular}{lccccccccccccc}
			\toprule
			Bob\textbackslash Eve & RL   & Random & Fixed(0.1) & Fixed(0.2) & Fixed(0.3) & Fixed(0.4) & Fixed(0.5) & Fixed(0.6) & Fixed(0.7) & Fixed(0.8) & Fixed(0.9) & Fixed(1.0) & Avg. \\
			\midrule
			RL
			& \textbf{\uline{0.42}}
			& \textbf{\uline{0.53}}
			& \textbf{\uline{3.37}}
			& \textbf{\uline{1.89}}
			& \textbf{\uline{1.22}}
			& \textbf{\uline{1.06}}
			& \textbf{\uline{0.86}}
			& \textbf{\uline{0.81}}
			& \textbf{\uline{0.70}}
			& \textbf{\uline{0.60}}
			& \textbf{\uline{0.49}}
			& \textbf{\uline{0.45}}
			& \textbf{\uline{1.03}} \\
			Random      & 0.25 & 0.27 & 0.37 & 0.35 & 0.31 & 0.30 & 0.27 & 0.24 & 0.21 & 0.18 & 0.16 & 0.14 & 0.25 \\
			\midrule
			Fixed(0.1)  & 0.00 & 0.25 & 1.91 & 1.46 & 0.69 & 0.33 & 0.00 & 0.00 & 0.00 & 0.00 & 0.00 & 0.00 & 0.39 \\
			Fixed(0.2)  & 0.00 & 0.46 & 0.97 & 0.95 & 0.86 & 0.70 & 0.59 & 0.37 & 0.31 & 0.00 & 0.00 & 0.00 & 0.44 \\
			Fixed(0.3)  & 0.19 & 0.48 & 0.65 & 0.65 & 0.64 & 0.60 & 0.54 & 0.47 & 0.41 & 0.33 & 0.27 & 0.18 & 0.45 \\
			Fixed(0.4)  & 0.29 & 0.42 & 0.49 & 0.48 & 0.48 & 0.48 & 0.46 & 0.42 & 0.39 & 0.36 & 0.32 & 0.26 & 0.41 \\
			Fixed(0.5)  & 0.29 & 0.37 & 0.39 & 0.39 & 0.39 & 0.39 & 0.38 & 0.37 & 0.36 & 0.33 & 0.31 & 0.28 & 0.35 \\
			Fixed(0.6)  & 0.28 & 0.32 & 0.32 & 0.32 & 0.33 & 0.32 & 0.32 & 0.32 & 0.31 & 0.30 & 0.29 & 0.27 & 0.31 \\
			Fixed(0.7)  & 0.27 & 0.27 & 0.28 & 0.28 & 0.28 & 0.28 & 0.28 & 0.28 & 0.27 & 0.26 & 0.26 & 0.24 & 0.27 \\
			Fixed(0.8)  & 0.24 & 0.23 & 0.24 & 0.24 & 0.24 & 0.23 & 0.23 & 0.23 & 0.22 & 0.22 & 0.22 & 0.21 & 0.23 \\
			Fixed(0.9)  & 0.20 & 0.19 & 0.20 & 0.20 & 0.20 & 0.19 & 0.19 & 0.19 & 0.18 & 0.18 & 0.18 & 0.17 & 0.19 \\
			Fixed(1.0)  & 0.17 & 0.16 & 0.17 & 0.16 & 0.16 & 0.16 & 0.16 & 0.15 & 0.15 & 0.15 & 0.15 & 0.14 & 0.16 \\
			\midrule
			Col. Avg.
			& \textbf{\uline{0.22}}
			& \textbf{\uline{0.33}}
			& \textbf{\uline{0.78}}
			& \textbf{\uline{0.56}}
			& \textbf{\uline{0.43}}
			& \textbf{\uline{0.38}}
			& \textbf{\uline{0.31}}
			& \textbf{\uline{0.28}}
			& \textbf{\uline{0.25}}
			& \textbf{\uline{0.21}}
			& \textbf{\uline{0.19}}
			& \textbf{\uline{0.18}}
			& --- \\
			\bottomrule
	\end{tabular}}
	\vspace{-10pt}
\end{table*} 
\begin{figure*}[t]
	\centering
	\begin{minipage}{0.31\textwidth}
		\centering
		\includegraphics[width=\linewidth]{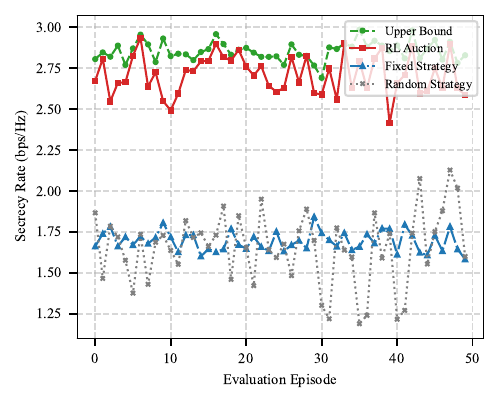}
		\caption{SR vs. Episode}
		\label{fig:SA_Trace}
	\end{minipage} \hfill
	\begin{minipage}{0.31\textwidth}
		\centering
		\includegraphics[width=\linewidth]{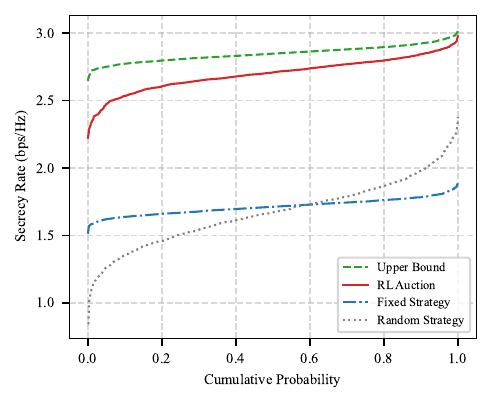}
		\caption{SR vs. CDF}
		\label{fig:SA_CDF}
	\end{minipage} \hfill
	\begin{minipage}{0.31\textwidth}
		\centering
		\includegraphics[width=\linewidth]{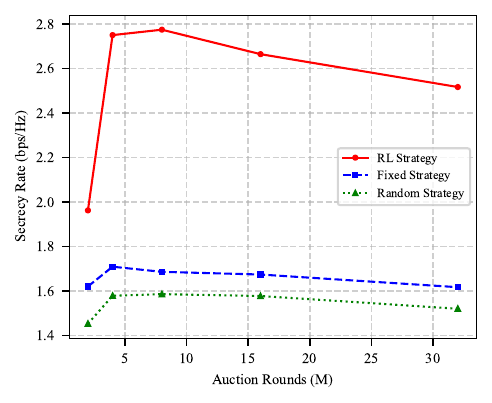}
		\caption{SR vs. Number of Auction Rounds}
		\label{fig:Robust_M}
	\end{minipage}
	\vspace{-20pt}
\end{figure*}
 
\begin{figure*}[t]
	\centering
	\begin{minipage}{0.31\textwidth}
		\centering
		\includegraphics[width=\linewidth]{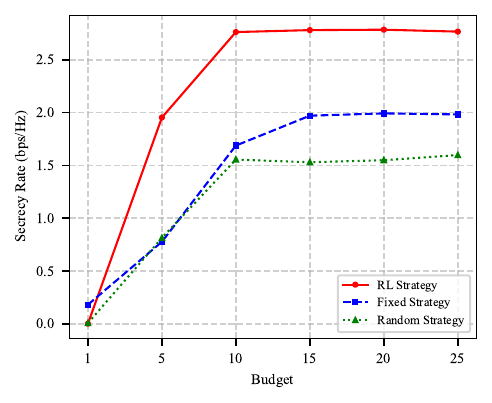}
		\caption{SR versus the Initial Budget}
		\label{fig:Robust_Budget}
	\end{minipage}\hfill
	\begin{minipage}{0.3\textwidth}
		\centering
		\includegraphics[width=\linewidth]{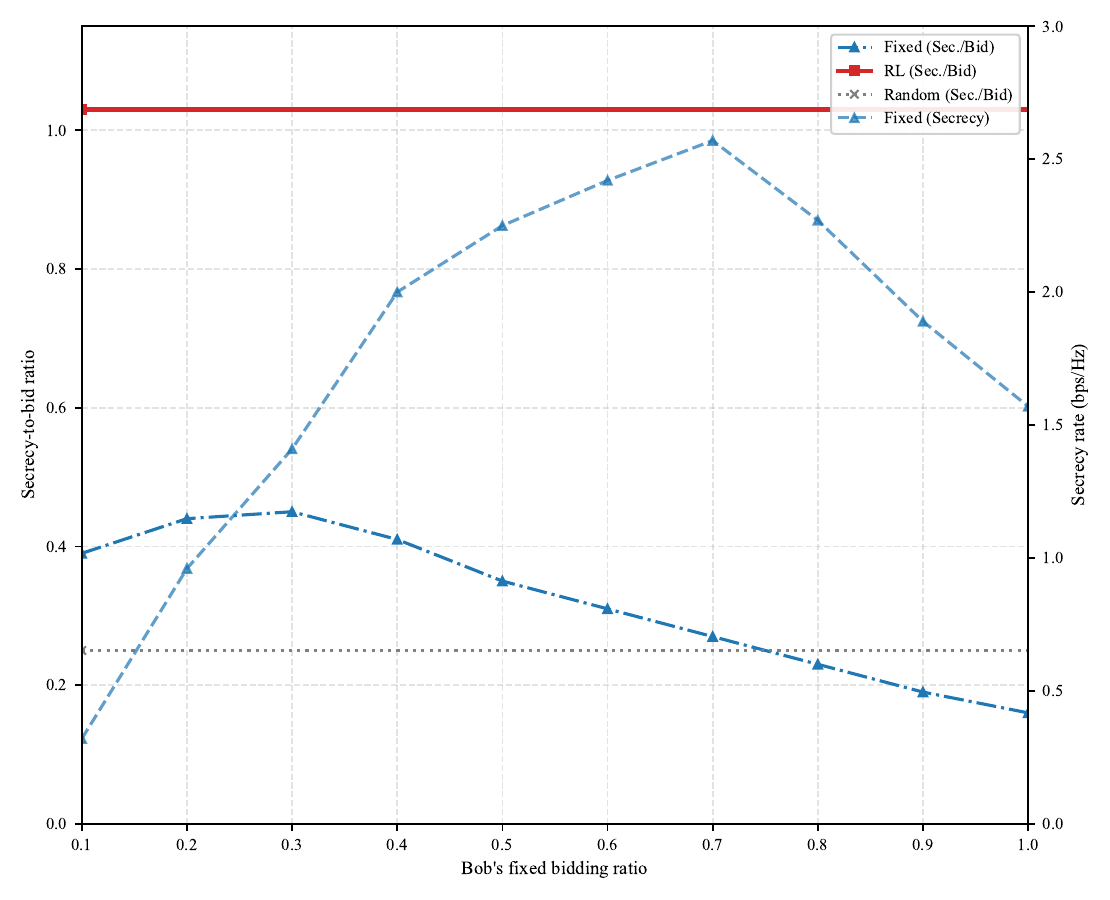}
		\caption{Unit-price secrecy versus Bob's bidding aggressiveness.}
		\label{fig:SecBid_ratio}
	\end{minipage}\hfill
	\begin{minipage}{0.31\textwidth}
		\centering
		\includegraphics[width=\linewidth]{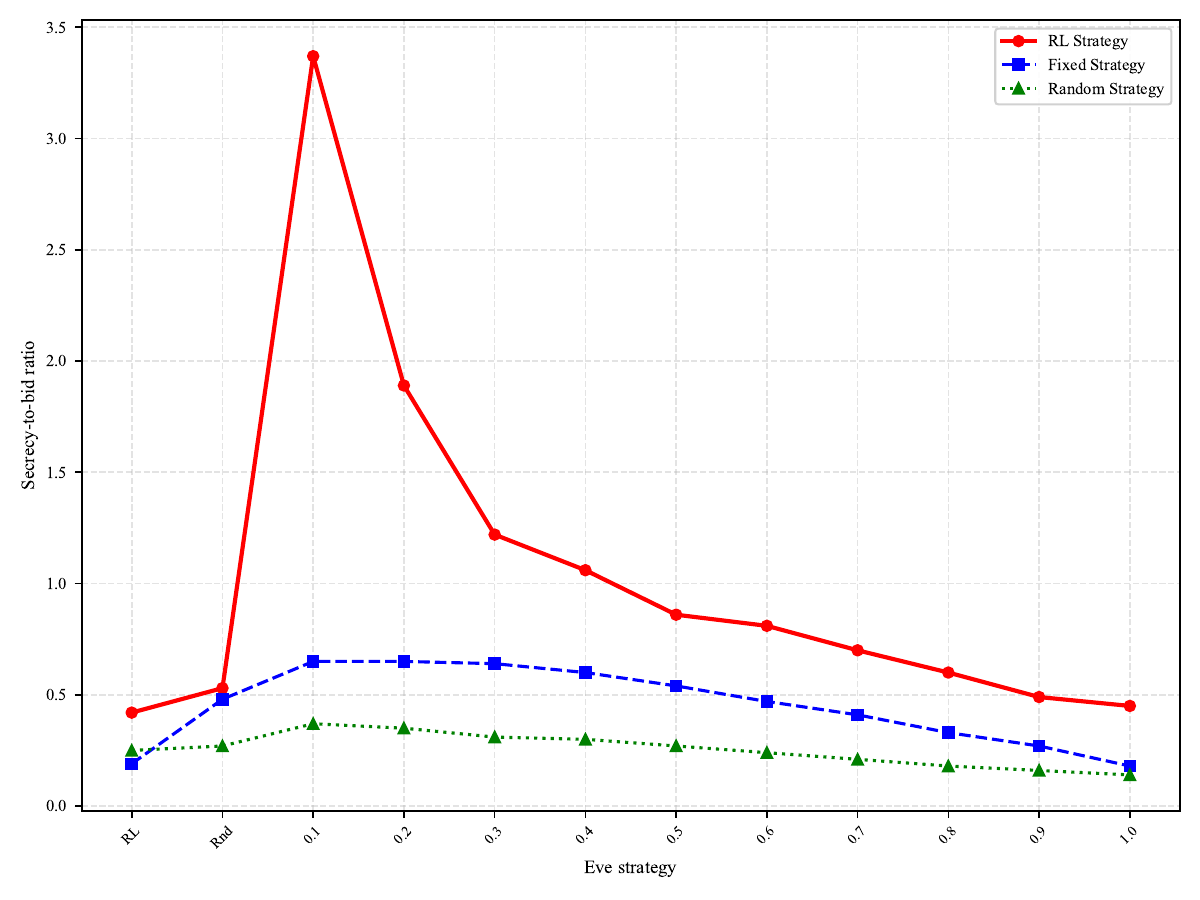}
		\caption{Unit-price secrecy of Bob versus Eve's strategy.}
		\label{fig:SecBid_eve}
	\end{minipage}
	\vspace{-15pt}
\end{figure*}

\begin{figure}[t]
	\centering
	\includegraphics[width=0.63\columnwidth]{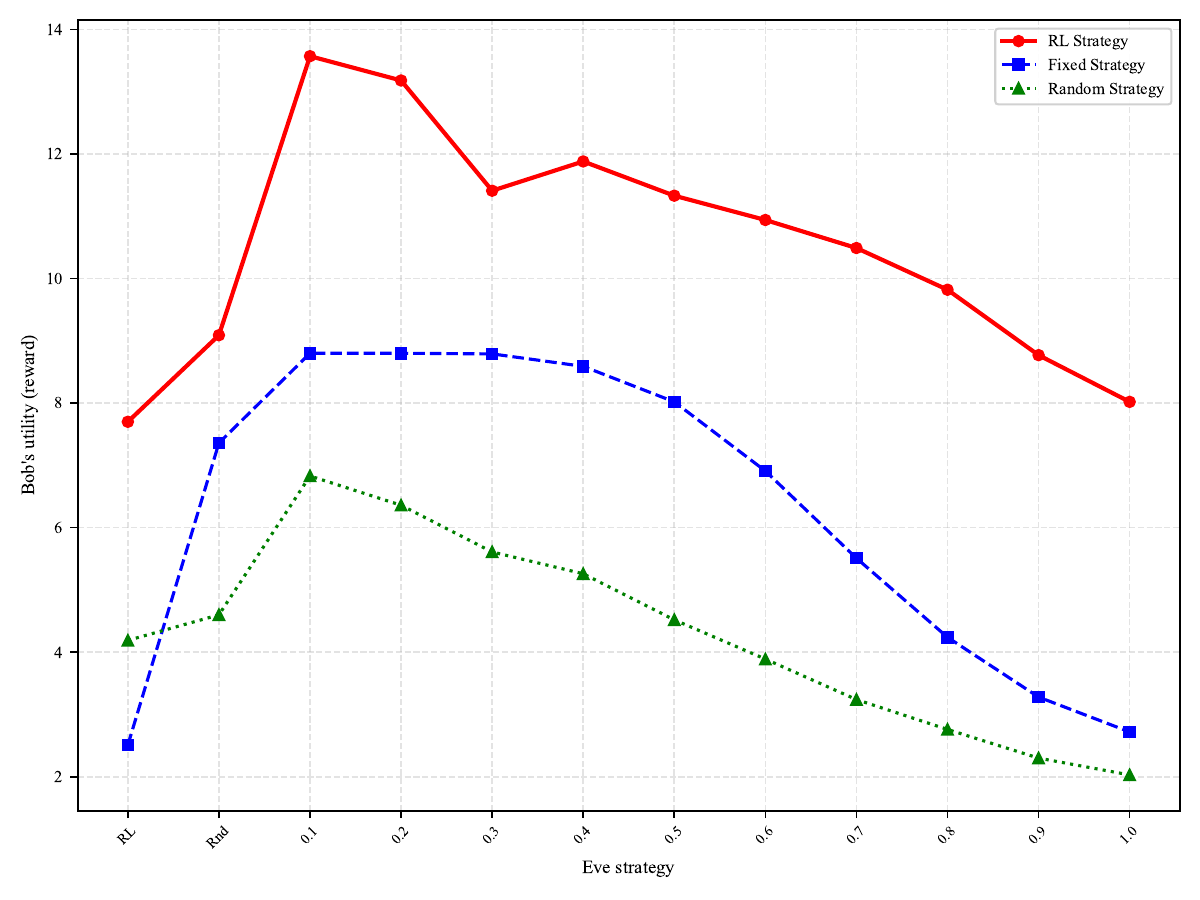}
	\caption{Bob's utility versus Eve's strategy.}
	\label{fig:Utility_eve}
	\vspace{-20pt}
\end{figure}

To quantify the secrecy loss induced by competitive RIS allocation, we compare the proposed MARL-based bidding strategy with an ideal centralized benchmark and two heuristic bidding strategies.
For each evaluation episode, all compared strategies are tested under the same channel realization, while the channel realizations are independently generated across episodes. 
Four representative reference scenarios are considered:
\textbf{1) Ideal Upper Bound:} Bob obtains all $M$ modules without competition, i.e., the ideal physical-layer upper bound under perfect global phase optimization.
\textbf{2) Random:} Bob bids stochastically without deterministic rules.
\textbf{3) Fixed:} Bob uses predetermined bidding ratios $0.1$-$1.0$.
\textbf{4)  RL Auction:} Bob uses the proposed policy, acting on the estimated CSI and the evolving auction state.

For the Random, Fixed, and RL cases, the reported secrecy rate is averaged over a heterogeneous adversary set (Eve~$\in\{$RL, Random, Fixed $0.1$--$1.0\}$), characterizing a highly non-cooperative and adaptive environment. In the trace and CDF results (Figs.~\ref{fig:SA_Trace} and~\ref{fig:SA_CDF}), Fixed curve denotes the mean over Bob's fixed bidding ratios $0.1$-$1.0$.

Fig.~\ref{fig:SA_Trace} depicts the episode-wise secrecy rate trace of the secrecy rate. The proposed RL-based auction consistently tracks the ideal physical-layer upper bound, indicating that the learned policy prioritizes modules with high long-term value under the evolving budget and competitive state, which indirectly improves secrecy performance. In contrast, the fixed strategy fluctuates and degrades because it cannot adapt to time-varying channels and inter-temporal budget constraints, while the random strategy is far inferior with severe fluctuations, indicating unreliable network. Although the RL curve may occasionally dip under budget or adversarial pressure, the sequential-auction framework enables continuous adaptation and thus prevents persistent performance loss.

Fig.~\ref{fig:SA_CDF} further characterizes the statistical performance using the cumulative distribution function (CDF) of the secrecy rate, where a clear dominance of the RL-based strategy is observed. The RL Auction consistently shifts the distribution toward higher secrecy rates, maintaining a range approximately between 2.2 and 3.0 bps/Hz and achieving a median value close to the ideal physical-layer upper bound. In contrast, the random strategy exhibits pronounced outage behavior, with a substantial probability mass concentrated at low secrecy rates, while the fixed strategy improves stability but remains confined to a narrow and suboptimal performance range. Although a mild crossover appears in the upper tail, where random bidding occasionally benefits from favorable channel realizations, the RL approach maintains a significantly higher performance floor, thereby ensuring reliable secrecy even under adverse conditions.

Accordingly, the trace and CDF results confirm that the RL policy improves both the expected secrecy rate and the worst-case performance floor. These gains stem from the bidder's ability to capture the inter-temporal value of RIS modules and to bid in a forward-looking, budget-aware manner, verifying that introducing economic competition into RIS allocation does not compromise physical-layer security.

\subsection{Secrecy Sensitivity to Auction and Budget Constraints}

This subsection investigates the impact of the key parameters of the system and the economy on the secrecy performance under the proposed sequential auction mechanism. Rather than placing emphasis on the robustness of the algorithm, the objective of this analysis is to elucidate the manner in which the scale of resources, the horizon of allocation, and the constraints of the budget jointly shape the resulting performance of security through strategic competition across multiple rounds.

\subsubsection{Impact of Auction Rounds}

Fig.~\ref{fig:Robust_M} shows that the RL strategy reaches its peak secrecy within a small number of auction rounds and retains a clear margin over the baselines thereafter, confirming that a few sequential rounds suffice to establish an effective allocation pattern. Changing the number of rounds also changes the auction granularity in the modularized RIS setting: coarse partitioning limits allocation flexibility, whereas overly fine partitioning fragments both marginal gains and bidding budgets. This exposes a granularity-budget trade-off, within which the RL policy remains the most robust.

\subsubsection{Impact of Budget Constraints}

Fig.~\ref{fig:Robust_Budget} reveals a threshold-and-saturation behavior. At very low budgets, all methods are scarcity-limited and perform poorly, and RL may even trail simpler ones because its strategic advantage cannot be activated under extreme scarcity. Once the budget enters a strategy-sensitive regime, RL improves sharply and dominates the baselines because the bidder must balance current wins against future opportunities. As the budget grows further, all methods become resource-limited and saturate, yet RL retains a consistent advantage, indicating that its superiority stems from selecting high-value configurations rather than from resource abundance alone.

\subsection{Cost-Aware Secrecy Efficiency}

The Bid column of Table~\ref{tab:SA_summary} reports the row-average bid of Bob under different strategies. The RL-based strategy adopts a moderate, adaptive bidding pattern, avoiding both excessive aggressiveness and extreme conservatism, whereas fixed strategies follow rigid rules that either rapidly deplete the budget or fail to compete. On average, RL sustains a moderate bid of $3.55$, far below the aggressive fixed strategies (up to $9.91$) yet sufficient to secure high-value modules, which avoids overbidding and preserves budget for high-impact rounds.

This advantage is quantified by the secrecy-to-bid ratio in Table~\ref{tab:SA_rate_bid_ratio}, where RL attains far higher cost efficiency than every baseline. Fixed strategies suffer from either underbidding, which loses critical modules, or overbidding, which yields diminishing returns, whereas RL sustains a stable efficiency even under intense adversarial pressure. As visualized in Figs.~\ref{fig:SecBid_ratio} and~\ref{fig:SecBid_eve}, Bob's secrecy-to-bid ratio falls from $0.45$ to $0.16$ as its fixed ratio grows, while RL sustains $1.03$ and attains the highest ratio against every considered Eve strategy, confirming that secrecy-to-bid efficiency is maximized by RL under the considered adversarial strategies.

\vspace{-7pt}
\subsection{Economic Performance and RIS Provider Revenue}

The RIS Rev. column of Table~\ref{tab:SA_summary} reports the RIS provider's row-average revenue. Provider revenue generally increases with bidding aggressiveness because stronger competition raises transaction prices. In contrast, the RL-based strategy maintains moderate revenue by avoiding excessive payments and optimizing long-term bidder utility. Aggressive fixed or random bidding may generate higher revenue, but does not necessarily improve secrecy or allocation efficiency. Hence, revenue-maximizing outcomes are not necessarily security-efficient, whereas the RL strategy achieves a better balance between secrecy performance and payment cost.
Fig.~\ref{fig:Utility_eve} further shows Bob's utility under different Eve strategies. 
The RL strategy maintains the highest utility across the considered adversarial strategies, confirming that its secrecy advantage is achieved without excessive payment.

Overall, secrecy performance under the SRA mechanism is jointly shaped by physical-layer module values, budget constraints, and bidding strategies. Neither increasing the auction horizon nor enlarging the budget guarantees monotonic secrecy gains, as both alter the temporal allocation of resources and may introduce diminishing returns. By adapting bids to module valuations, competitive conditions, and remaining budgets, the RL strategy consistently outperforms non-adaptive baselines, demonstrating the importance of strategic bidding in translating RIS resources into secrecy gains.

\section{Conclusion}\label{sec_con}

This paper studies RIS-assisted physical-layer security from a market-oriented perspective, where a legitimate receiver and a strategic eavesdropper compete for limited RIS control rights through sequential auctions.
By integrating state-dependent module valuation with MARL-based bidding, the proposed SRA framework captures budget-coupled strategic interactions and enables forward-looking resource acquisition under incomplete information. 
Simulation results show that the learned strategy improves secrecy performance and bidding utility over non-adaptive baselines, revealing that security in shared RIS markets is jointly determined by physical-layer conditions and economic allocation decisions.

\bibliographystyle{IEEEtran}
\bibliography{reference}

\end{document}